\documentclass[11pt]{article}

\usepackage{amsmath, amsthm}

\theoremstyle{definition}
\newtheorem{theorem}{Theorem}[section]

\newtheorem{lemma}[theorem]{Lemma} 
\newtheorem{definition}[theorem]{Definition} 

\newcommand{\NP}{{\textsf{NP}}}
\newcommand{\PO}{{\textsf{P}}}
\newcommand{\coNP}{{\textsf{coNP}}}
\newcommand{\FP}{{\textsf{FP}}}
\newcommand{\DP}{{\textsf{DP}}}
\newcommand{\comments}[1]{}

\title{Computational Complexity of Quadratic Unconstrained Binary Optimization} %TODO Please add

\author{Hirotoshi Yasuoka\\hirotoshi.yasuoka@gmail.com}

\begin{document}
\maketitle

\begin{abstract}
In this paper, we study the computational complexity of the quadratic unconstrained binary optimization (QUBO) problem under the functional problem ${\FP}^{\NP}$ categorization.  We focus on four sub-classes: (1) When all coefficients are integers QUBO is ${\FP}^{\NP}$-complete.  (2) When every coefficient is an integer lower bounded by a constant $k$,  QUBO is ${\FP}^{\NP[\log]}$-complete.  (3) When every coefficient is an integer upper bounded by a constant $k$,  QUBO is again ${\FP}^{\NP[\log]}$-complete.  (4) When coefficients can only be in the set $\{1, 0, -1\}$, QUBO is ${\FP}^{\NP[\log]}$-complete.  With recent results in quantum annealing able to solve QUBO problems efficiently, our results provide a clear connection between quantum annealing algorithms and the ${\FP}^{\NP}$ complexity class categorization.  We also study the computational complexity of the decision version of the QUBO problem with integer coefficients.  We prove that this problem is ${\DP}$-complete.

\end{abstract}
\section{Introduction}

The quadratic unconstrained binary optimization (QUBO) problem is the problem, given a quadratic expression over binary variables where the coefficients are integers, to find the minimum value of the expression.

The QUBO problem can be solved by quantum annealing which was proposed by Kadowaki and Nishimori~\cite{PhysRevE.58.5355}.   Quantum annealing is a method that solves optimization problems by utilizing quantum fluctuation.  D-Wave Systems, Inc. released a quantum computer that realizes quantum annealing.    The quantum computers that realize quantum annealing can efficiently find a solution to the QUBO problem~\footnote{For example, Chapuis et al.~\cite{DBLP:journals/vlsisp/ChapuisDHR19} mentioned that a quantum computer, the D-Wave 2X quantum annealer, can solve the QUBO problem with 45 binary variables in 0.06 seconds.}.  Therefore, researchers are investigating potential applications of QUBO (such as machine learning~\cite{10.3389/fphy.2018.00055}, traffic flow optimization~\cite{neukart2017traffic}, and automated guided vehicle control in a factory~\cite{10.3389/fcomp.2019.00009}).  

In this paper, we investigate the hardness of QUBO problems by using computational complexity theory.  Our results clarify that which problems can be solved via QUBO and which can not, therefore our results provide insight into which applications can receive a benefit from quantum annealing and which can not.  Specifically, we prove that QUBO is ${\FP}^{\NP}$-complete.  The computational complexity class ${\FP}^{\NP}$ is the set of functional problems that are solvable by a deterministic polynomial-time Turing machine that can query an oracle for an ${\NP}$-complete problem.  The completeness result implies that every problem in ${\FP}^{\NP}$ can be solved via QUBO.  For example, the traveling salesman problem is an ${\FP}^{\NP}$-complete problem, which was proven by Krentel~\cite{DBLP:journals/jcss/Krentel88}.  Intuitively, our result indicates that QUBO is as hard as the traveling salesman problem.   The completeness result also implies that a problem that is harder than ${\FP}^{\NP}$ cannot be solved via QUBO.   For example, 2QBF is an ${\NP}^{\NP}$-complete problem.  2QBF is the problem of checking whether a given quantified Boolean formula $\exists \vec{x} \forall\vec{y} .\phi$ where $\phi$ is a quantifier free Boolean formula can be satisfied.  It is unlikely that 2QBF can be solved via QUBO, as ${\NP}^{\NP}$ is believed to be harder than ${\FP}^{\NP}$~\footnote{This topic is discussed in more detail in Section~\ref{ssec:app}.}.

We also prove the hardness of subsets of the QUBO problem as follows:
\begin{itemize}
    \item QUBO is ${\FP}^{{\NP}[\log]}$-complete when every coefficient is greater than some constant bound.
    \item QUBO is ${\FP}^{{\NP}[\log]}$-complete when every coefficient is less than some constant bound.
    \item QUBO is ${\FP}^{{\NP}[\log]}$-complete when every coefficient is $-1$, $0$, or $1$.
\end{itemize}
The computational complexity class ${\FP}^{{\NP}[\log]}$ is the class of functional problems that are solvable by a deterministic polynomial-time Turing machine that can query an oracle for an ${\NP}$-complete problem $O(\log M)$ times where $M$ is the size of the input of the Turing machine.   For example, the clique problem is a ${\FP}^{{\NP}[\log]}$-complete problem.  The clique problem is the problem of finding the size of the largest clique in a given undirected graph.   Intuitively, our results show that QUBO with constant lower bounds for the coefficients is as hard as the clique problem.  

We also prove the hardness of the decision version of the QUBO problem, that is, the problem of checking whether the minimum value of a given quadratic binary expression with integer coefficients is equal to a given integer.  The problem is a ${\DP}$-complete problem.  The computational complexity class ${\DP}$ is the set of decision problems such that ${\DP}=\{P_0\cap P_1\mid P_0\in{\NP}, P_1\in{\coNP}\}$, which is defined by Papadimitriou and Yannakakis~\cite{PAPADIMITRIOU1984244}.  For example, the exact clique problem is a ${\DP}$-complete problem.  The exact clique problem is the problem of deciding whether the size of the largest clique in a given undirected graph is equal to a given integer.  Intuitively, our result show that the decision version of the QUBO problem is as hard as the exact clique problem.  

The rest of this paper is organized as follows.  Section~\ref{sec:pre} presents the problem definitions and the computational complexity class used in proving the hardness of QUBO problems.   Section~\ref{sec:qc} proves that QUBO is an ${\FP}^{\NP}$-complete problem, while Section~\ref{sec:bqc} demonstrates the complexity class of QUBO with constant bounds for the coefficients.  Section~\ref{sec:dqubo} demonstrates the complexity class of the decision version of QUBO problem.  Section~\ref{sec:dis} discusses the implications of our results, while Section~\ref{sec:rw} discusses related work.  Finally, Section~\ref{sec:con} concludes the paper.  

\section{Preliminaries}\label{sec:pre}
First, we formally define the QUBO problem as follows.
\begin{definition}[QUBO]
Let $x_1, x_2, ..., x_n$ be binary variables, and $Q$ be an  $n \times n$ upper triangular matrix of integers.  The QUBO problem is the problem of minimizing the following quadratic expression:
\[
\sum_{i=1}^n \sum_{j=1}^n q_{ij} x_i x_j
\]
\end{definition}
It is worth noting that quantum annealing can solve QUBO that has a quadratic expression whose coefficients are real numbers; however, in the problem above, the coefficients are integers~\footnote{We discuss the computational complexity of QUBO that has a quadratic expression whose coefficients are rational numbers.  See Section~\ref{ssec:hqrc}.}.  That is, this paper only investigates the computational complexity of a subset of the problem that is solvable by quantum annealing. 

Next, we describe the computational complexity classes that we use in proving the hardness of QUBO.  We assume that readers have basic knowledge of computational complexity theory~\cite{sipser13}; thus, we omit explanations of $\PO$ and $\NP$.

Recall that the computational complexity class ${\FP}^{\NP}$ is the set of functional problems that is solvable by a deterministic polynomial-time Turing machine that can query an oracle for an ${\NP}$-complete problem.   ${\FP}^{{\NP}[\log]}$ is the set of functional problems that is solvable by a deterministic polynomial-time Turing machine that can query an oracle for an ${\NP}$-complete problem $O(\log M)$ times, where $M$ is the size of the input of the Turing machine.  ${\DP}$ is the set of decision problems such that ${\DP}=\{P_0\cap P_1\mid P_0\in{\NP}, P_1\in{\coNP}\}$

In this paper, we use metric reduction to analyze the relation between functional problems, as this method has been used in previous research to analyze the hardness of optimization problems~\cite{DBLP:journals/jcss/Krentel88, DBLP:journals/mst/GasarchKR95}.  Metric reduction is defined as follows.
\begin{definition}[Metric reduction~\cite{DBLP:journals/jcss/Krentel88}]
Let $P_0$ and $P_1$ be functions.  A metric reduction from $P_0$ to $P_1$ is a pair of polynomial-time computable functions $(f,g)$ such that
$P_0(x)=g(P_1(f(x)),x)$ for all $x$.
\end{definition}
Intuitively, $f$ transforms the input of $P_0$ to the input of $P_1$, and $g$ transforms the output of $P_1$ to $P_0$ such that the transformed output is the actual solution of $P_0$.

\section{Hardness of QUBO}\label{sec:qc}
In this section, we demonstrate the computational complexity of QUBO.  First, we prove that QUBO is an ${\FP}^{\NP}$-hard problem by reducing 0-1 integer linear programming (01IP) to QUBO.  01IP is defined as follows.
\begin{definition}[01IP]
Given integer matrix $A$ and integer vectors $B$ and $C$,  find the maximum value of $C^T X$ over all binary vector $X$ subject to $A X \le B$. 
\end{definition}
\noindent
01IP is an ${\FP}^{\NP}$-complete problem under metric reduction, which was proven by Krentel~\cite{DBLP:journals/jcss/Krentel88}.  

\begin{theorem}\label{thm:qfpnph}
QUBO is ${\FP}^{\NP}$-hard.
\end{theorem}
\begin{proof}
We demonstrate a method to transform the instance of 01IP to the instance of QUBO in polynomial time in the size of the instance of 01IP such that the solution of the instance of 01IP is equal to the product of the solution of the instance of QUBO and $-1$.   Note that the pair of the method to transform the instance of 01IP into the instance of QUBO and the function that multiplies the input by $-1$ is the metric reduction from 01IP to QUBO.

Let $A$ be an $m\times n$ integer matrix, and $B$ and $C$ be integer vectors of the following form.
\[
A = 
\begin{bmatrix} 
a_{1 1}&a_{1 2}& ...& a_{1 n}\\
a_{2 1}&a_{2 2}& ...& a_{2 n}\\
\vdots&&&\vdots\\
a_{m 1}&a_{m 2}& ...& a_{m n}
\end{bmatrix}
\;
B = 
\begin{bmatrix}
b_1\\b_2\\\vdots\\b_m
\end{bmatrix}
\;
C = 
\begin{bmatrix}
c_1\\c_2\\\vdots\\c_n
\end{bmatrix}
\]
We wish to solve the following instance of 01IP:
\[
\max_{X} \;C^T X \;\;\;
s.t. \; AX\le B 
\]

For this instance, we construct the following quadratic binary expression:  
\[
\mathrm{QB}(X,Y) = \mathrm{Obj}(X) + h \cdot \mathrm{Cstr}(X,Y),
\]
where 
\[
\begin{array}{rcl}
\mathrm{Obj}(X) &=& -\sum_{i\in\{1,..,n\}} c_i (x_i)^2\\
\mathrm{Cstr}(X,Y) &=& \sum_{i\in\{1,..,m\}}\left(\sum_{j\in\{1,..,n\}} a_{i j} x_j  + y_i - b_i\right)^2\\
y_i&=& \sum_{\ell\in\{1,..,k\}} 2^{\ell-1} y_{i \ell}\\
Y &=& \begin{bmatrix}y_{1 1}& \dots & y_{1 k}\\
y_{2 1}&\dots& y_{2 k}\\
&\vdots&\\
y_{m 1}&\dots& y_{m k}
\end{bmatrix}\\
h&=& n\cdot \max_i c_i +1\\
k&=&\lfloor \log(\max_i b_i -n\cdot \min_{i,j}a_{i j}) \rfloor + 1
\end{array}
\]
We provide an intuitive explanation of $\mathrm{QB}(X,Y)$.  Binary vector $X$ represents the binary vector appearing in the instance of 01IP, while binary matrix $Y$ is used to transform the inequality constraints of the instance of 01IP into equality constraints.  For example, the $i$-th column of the inequality constraint $AX\le B$ is transformed into the following equality. 
\[
\sum_{j\in\{1,..,n\}} a_{i j} x_j  +y_i - b_i =0
\]
where $y_i$ is a slack variable represented by binary representation, that is, $y_i = \sum_{\ell\in\{1,..,k\}} 2^{\ell-1} y_{i \ell}$.  The term $\mathrm{Cstr}(X,Y)$ adds a penalty to $\mathrm{QB}(X,Y)$ when $X$ does not satisfy the constraints of the instance of 01IP.  The term $\mathrm{Obj}(X)$ adds a penalty to $\mathrm{QB}(X,Y)$ when $C^T X$ decreases.  Therefore, by minimizing $\mathrm{QB}(X,Y)$,  the obtained assignment of $X$ is the solution of the instance of 01IP.  The coefficient $h$ is selected to add sufficiently large penalty when $\mathrm{Cstr}(X,Y)\not=0$.  Thus, the binary vector $X$ and the binary matrix $Y$ that minimize $\mathrm{QB}(X,Y)$ will satisfy $\mathrm{Cstr}(X,Y)=0$.    

We show that the minimum value of $\mathrm{QB}(X,Y)$ is equal to the product of the solution of the instance of 01IP defined above and $-1$.  We use the following propositions to prove the theorem.
\begin{itemize}
    \item If $\mathrm{Cstr}(X,Y)=0$ for some $Y$ then $A X\le B$.
    \item If $\mathrm{QB}(X,Y)$ attains the minimum value when $X=X'$ and $Y=Y'$, then $\mathrm{Cstr}(X',Y')=0$.
\end{itemize}

We prove the first proposition by proving its contrapositive.  Suppose that $AX\not\le B$, that is, there exist $\vec{x}=x_1,\dots,x_n$ and $i$ such that $\sum_{j\in\{1,\dots,m\}} a_{i j}x_j - b_i >0$.  We have 
\[
\begin{array}{l}
\sum_{j\in\{1,\dots,m\}} a_{i j}x_j - b_i>0\\
\;\;\Rightarrow \sum_{j\in\{1,..,n\}} a_{i j} x_j  \\
\;\;\;\;\;\;+\sum_{l\in\{1,..,k\}} 2^{l-1} y_{i l} - b_i > 0 \;\;\textrm{for any}\; Y\\
\;\;\Rightarrow \textrm{Cstr}(X,Y) \not= 0\;\;\textrm{for any}\;Y
\end{array}
\]

Next, we prove the second proposition by proving its contrapositive.  Suppose that there exist $X'$ and $Y'$ such that $\mathrm{Cstr}(X',Y')\not=0$.  Then, we have
\[
\mathrm{QB}(X',Y')\ge \textrm{Obj}(X') + h\ge -n\cdot\max_i c_i + h
\]
Let $Y_B$ be the binary matrix such that $Y_B \begin{bmatrix}2^0& 2^1& \dots& 2^{k-1}\end{bmatrix}^T - B = 0$.  Note that $\mathrm{Cstr}(\vec{0}, Y_B)=0$.  Thus, we have
\[
\begin{array}{rcl}
\mathrm{QB}(X',Y')&\ge& -n\cdot\max_i c_i + h \\
&=& 1\\
&>& 0 = \mathrm{Obj}(\vec{0})+\mathrm{Cstr}(\vec{0},Y_B)
\end{array}
\]
It follows that $\mathrm{QB}(X',Y')$ is not a minimum value.

Using the two propositions above, we show that the minimum value of $\mathrm{QB}(X,Y)$ is equal to the product of the solution of the instance of 01IP and $-1$.  Suppose that $\mathrm{QB}(X,Y)$ attains the minimum value when $X=X'$ and $Y=Y'$.  In addition, suppose that $\vec{x}$ is the solution to the instance of 01IP.  We have $\mathrm{QB}(X',Y')\le -C^T \vec{x}$, because there exists $Y_{\vec{x}}$ such that $\mathrm{Cstr}(\vec{x},Y_{\vec{x}}) =0$ and $-C^T \vec{x}=\mathrm{QB}(\vec{x},Y_{\vec{x}})\ge \mathrm{QB}(X',Y')$.  Therefore, it suffices to show that $\mathrm{QB}(X',Y')\ge -C^T\vec{x}$.  By the fact that $\mathrm{QB}(X',Y')$ is the minimum value and by the two propositions above, we have $\mathrm{Cstr}(X',Y')=0$ and hence $A X' \le B$.  This implies that $X'$ is a possible solution of the instance of 01IP.  Thus, we have $-C^T \vec{x} \le -C^T X' = \mathrm{QB}(X',Y')$.

Finally, we show that $\mathrm{QB}(X,Y)$ can be obtained in polynomial time in the size of the given instance of 01IP.  The number of binary variables in $\mathrm{QB}(X,Y)$, is $n+m k$ where $k =\lfloor \log(\max_i b_i - n\cdot \max_{i j} a_{i j}) \rfloor + 1$.  $n+m k$ is bounded by $O(M^2)$ where $M$ is the size of the instance of 01IP.  This is because $b_i$ and $a_{i j}$ are bounded by $2^M$, and $k$ is bounded by $O(M)$.   Each coefficient of $\mathrm{QB}(X,Y)$ can be obtained in polynomial time.  Therefore, $\mathrm{QB}(X,Y)$ can be obtained in polynomial time in the size of the instance of 01IP.
\end{proof}

Next, we prove that QUBO is in ${\FP}^{\NP}$.   We formulate the decision problem, DLEQUBO, which is in $\NP$.  Then, we show that QUBO can be solved in polynomial time by querying the oracle for DLEQUBO.  
\begin{definition}[DLEQUBO]
Let $X$ be a vector of binary variables, $Q$ be an $n \times n$ upper triangular matrix of integers, and $q$ be an integer.  DLEQUBO is the problem of deciding whether the following inequality holds:
\[
\min_X \sum_{i=1}^n \sum_{j=1}^n q_{ij} x_i x_j \le q 
\]
\end{definition}
Intuitively, DLEQUBO checks whether the solution of the instance of QUBO is less than or equal to a given integer.   We prove that DLEQUBO is in $\NP$ below.
\begin{lemma}\label{lem:DLEQUBO}
DLEQUBO is in \NP.
\end{lemma}
\begin{proof}
We have
\begin{align*}
\min_X \sum_{i=1}^n \sum_{j=1}^n q_{ij} x_i x_j \le q \Leftrightarrow \bigvee_X \left(\sum_{i=1}^n \sum_{j=1}^n q_{ij} x_i x_j\le q\right) 
\end{align*}
Trivially, the disjunction of the inequalities above is solvable by a non-deterministic Turing machine in polynomial time.  Therefore, we have that DLEQUBO is in \NP. 
\end{proof}
By using Lemma~\ref{lem:DLEQUBO}, we prove the following theorem by demonstrating that QUBO is solvable in polynomial time by querying the oracle for DLEQUBO.  
\begin{theorem}\label{thm:qfpnp}
QUBO is in ${\FP}^{\NP}$.
\end{theorem}
\begin{proof}
We show that the solution of QUBO can be obtained in polynomial time in the size of the quadratic expression of a given QUBO problem by querying the oracle for DLEQUBO.   Let $lb$ be an integer such that $lb = \sum \{q_{ij}\mid q_{ij}\le 0\}$.   Trivially, the solution of the given QUBO is in the set $[lb, 0]$.  We use binary search to find the solution of the given QUBO problem.   Our procedure queries the oracle for DLEQUBO $\log(-lb) + O(1)$ times.  We have $\log(-lb) + O(1) \le O(M)$, because $|lb|$ is bounded by $2^M$ where $M$ is the size of the quadratic  expression.   Therefore, QUBO is in $\FP^{\mbox{DLEQUBO}}$.  By Lemma~\ref{lem:DLEQUBO}, QUBO is in $\FP^\NP$.   
\end{proof}

We thus have the following completeness result of QUBO from Theorems~\ref{thm:qfpnph} and~\ref{thm:qfpnp}.  
\begin{theorem}\label{thm:qc}
QUBO is ${\FP}^{\NP}$-complete.
\end{theorem}

\section{Hardness of QUBO with a constant bound of coefficients}
\label{sec:bqc}
In this section, we demonstrate the computational complexity of QUBO with a constant bound of coefficients.  We analyze two problems.  One is QUBO with a constant lower bound, referred to as LQUBO, and the other is QUBO with a constant upper bound, referred to as UQUBO.  Formally, we define the problems as follows.
\begin{definition}[LQUBO]
Let $\ell$ be a constant such that $\ell<0$.  Let $X$ be a vector of binary variables, and $Q$ be an $n \times n$ upper triangular matrix of integers in which all elements are greater than $\ell$.  LQUBO is the problem of minimizing the following quadratic expression:
\[
\min_X \sum_{i=1}^n \sum_{j=1}^n q_{ij} x_i x_j
\]
\end{definition}
\begin{definition}[UQUBO]
Let $u$ be a constant such that $u>0$.  Let $X$ be a vector of binary variables, and $Q$ be an $n \times n$ upper triangular matrix of integers in which all elements are less than $u$.  UQUBO is the problem of minimizing the following quadratic expression:
\[
\min_X \sum_{i=1}^n \sum_{j=1}^n q_{ij} x_i x_j
\]
\end{definition}

First, we prove that LQUBO is in ${\FP}^{\NP[\log]}$.  Recall that ${\FP}^{{\NP}[\log]}$ is the class of functional problems solvable by a deterministic polynomial-time Turing machine that can query an oracle for an ${\NP}$-complete problem $O(\log M)$ times, where $M$ is the size of the input of the Turing machine.  In contrast to QUBO, LQUBO can be solved by querying DLEQUBO a logarithmic number of times in the size of the quadratic expression.  This is because the Turing machine that solves LQUBO uses binary search over the possible solutions where the number of possible solutions of LQUBO is bounded by a polynomial of the size of the quadratic expression.   
\begin{theorem}\label{thm:lqin}
LQUBO is in ${\FP}^{\NP[\log]}$.
\end{theorem}
\begin{proof}
We show that the solution of LQUBO with an $n\times n$ integer matrix and a constant $\ell$ can be obtained in polynomial time by querying the oracle for DLEQUBO $O(\log M)$ times where $M$ is the size of the quadratic expression of a given LQUBO.   Trivially, the solution of the given LQUBO is in the set $[\ell \cdot n^2, 0]$.  We use a binary search to find the solution of the given LQUBO.   Our procedure queries the oracle for DLEQUBO $\log(-\ell\cdot n^2) + O(1)$ times.  We have $\log(-\ell \cdot n^2) + O(1) = \log(-\ell) + 2\log(n) +O(1) \le O(\log M)$, because $-\ell$ is a constant, and $n$ is bounded by $M$.  Therefore, LQUBO is in $\FP^{\mbox{DLEQUBO}[\log]}$.  By Lemma~\ref{lem:DLEQUBO}, we have that LQUBO is in $\FP^{\NP[\log]}$.  
\end{proof}

Next, we prove that UQUBO is in ${\FP}^{\NP[\log]}$.  Recall that the coefficients of the quadratic expression of UQUBO are not lower bounded by a constant as with those of QUBO.  Therefore, we need to query the oracle for DLEQUBO $O(M)$ times to obtain the solution of UQUBO where $M$ is the size of the quadratic expression of a given UQUBO if we adopt the procedure in the proof of Theorem~\ref{thm:qfpnp}.  However, interestingly, we can reduce UQUBO to the problem of minimizing the quadratic expression whose coefficients are lower bounded, and this problem is solvable by querying the oracle for a problem in $\NP$ $O(\log M)$ times.
\begin{theorem}\label{thm:uqin}
UQUBO is in ${\FP}^{\NP[\log]}$.
\end{theorem}
\begin{proof}
We show the procedure which solves UQUBO with an $n\times n$ integer upper triangular matrix $Q$ and a constant upper bound $u$ in polynomial time by querying the oracle for a problem in $\NP$ $O(\log M)$ times where $M$ is the size of the quadratic expression of the given UQUBO.  

Let $A$ be an $n\times n$ integer matrix such that
\[
a_{i j}
=
\begin{cases}
q_{i j}& q_{i j} \ge -2(n-1)\cdot u\\
0& \text{otherwise}
\end{cases}
\]
and $B$ be an $n\times n$ integer matrix such that $B=Q-A$, and let $I_B=\{i,j\mid q_{i j}<-2(n-1)\cdot u\}$.  First, our procedure constructs the set $I_B$.  Next, it constructs the integer matrix $A$.  Finally, it calculates \[\min_{X \in \{X\mid \forall a\in I_B. x_a=1\}} \sum_{i=1}^n \sum_{j=1}^n a_{ij} x_i x_j + \sum \{q_{i j}\mid q_{i j} < -2(n-1)\cdot u\}\]

We show that the expression above is equal to the solution of the given UQUBO problem.  It suffices to show that if $\sum_{i}\sum_{j} q_{i j} x_i x_j$ attains the minimum value when $X= z_1, \dots, z_n$ then the binary vector $z_1, \dots, z_n$ satisfies $\forall a\in I_B. z_a=1$.
This is because, by using this proposition, we have
\[
\begin{array}{l}
\min_{X} \sum_{i=1}^n \sum_{j=1}^n q_{ij} x_i x_j\\
\;=\min_{X \in \{X\mid \forall a\in I_B. x_a=1\}} \sum_{i=1}^n \sum_{j=1}^n (a_{ij} + b_{ij}) x_i x_j\\
\; = \min_{X \in \{X\mid \forall a\in I_B. x_a=1\}} \sum_{i=1}^n \sum_{j=1}^n a_{ij} x_i x_j + \sum \{q_{i j}\mid q_{i j} < -2(n-1)\cdot u\}
\end{array}
\]   
  We prove its contrapositive.  Let $z_1, \dots, z_n$ be a  binary vector such that $\forall a\in I_B. z_a=1$ does not hold, that is, there exist $k$ and $\ell$ such that $z_k z_\ell=0$ where $q_{k \ell} < -2(n-1)\cdot u$.  Let $z_1', \dots, z_n'$ be a binary vector such that $z_k' z_\ell'=1$ and $z_i' = z_i$ for $i\not\in\{k,\ell\}$.  Then, we show that $\sum_{i}\sum_{j} q_{i j} z_i' z_j' < \sum_{i}\sum_{j} q_{i j} z_i z_j$ by case analysis.
\begin{itemize}
    \item[-] $k\not=\ell$

We have
\[
\begin{array}{l}
\sum_{i}\sum_{j} q_{i j} z_i z_j - \sum_{i}\sum_{j} q_{i j} z_i' z_j'\\
= \sum_{i\in \{i\mid 1\le i \le k\}} q_{i k} (z_i z_k - z_i' z_k') + \sum_{j\in\{j\mid j>k, j\not=\ell\}} q_{k j} (z_k z_j -z_k' z_j')\\
\;\;\;\;+  \sum_{i\in \{i\mid 1\le i \le\ell, i\not=k\}} q_{i \ell} (z_i z_\ell - z_i' z_\ell') + \sum_{j\in\{j\mid j>\ell\}} q_{\ell j} (z_\ell z_j - z_\ell' z_j') \\
\;\;\;\; + q_{k \ell} (z_k z_\ell - z_k' z_\ell')\\
\ge (n - 1)\cdot (-u) + (n - 1)\cdot (-u) -q_{k \ell}\\ 
\end{array}
\]
because for any $i$ and $j$ we have $z_i z_j - z_i' z_j' \le 0$, $q_{i j} \le u$, and hence $q_{i j} ( z_i z_j - z_i' z_j' ) \ge -u$.  By the fact that $q_{k\ell}<-2(n-1)\cdot u$, we have
\[
\begin{array}{l}
(n -1)\cdot (-u) + (n-1)\cdot (-u) -q_{k \ell}\\
= -u\cdot 2(n-1) -q_{k \ell} > 0
\end{array}
\]

    \item[-] $k=\ell$

We have
\[
\begin{array}{l}
\sum_{i}\sum_{j} q_{i j} z_i z_j - \sum_{i}\sum_{j} q_{i j} z_i' z_j'\\
= \sum_{i\in \{i\mid 1\le i < k\}} q_{i k} (z_i z_k - z_i' z_k') + \sum_{j\in\{j\mid j>k\}} q_{k j} (z_k z_j - z_k' z_j') \\
\;\;\;\;+q_{k k} (z_k z_k - z_k' z_k')\\
\ge (n-1)\cdot (-u) -q_{k k}\\
 > 0
\end{array}
\]
\end{itemize}
For any cases, $\sum_{i}\sum_{j} q_{i j} z_i z_j$ is not a minimum value.  Thus, we prove the proposition.

Next, we show that the solution of the given UQUBO can be obtained in polynomial time by querying the oracle for a problem in $\NP$ $O(\log M)$ times.  Because, obviously, our procedure can construct $I_B$ and $A$ and calculate $\sum \{q_{i j}\mid q_{i j}<-2(n-1)\cdot u\}$ in polynomial time, we show that it can find the solution of $\min_{X \in \{X\mid \forall a\in I_B. x_a=1\}} \sum_{i=1}^n \sum_{j=1}^n a_{ij} x_i x_j$ by querying the oracle for a problem in $\NP$ $O(\log M)$ times.  It queries the oracle for the problem of deciding whether the following inequality holds.   
\[
\min_{X\in \{X\mid \forall a\in I_B.x_a=1\}} \sum_{i=1}^n \sum_{j=1}^n a_{ij} x_i x_j \le q 
\]
where $q$ is an integer.  This problem is in $\NP$, because the inequality above is equivalent to the following disjunction of inequalities which is solvable by a non-deterministic Turing machine in polynomial time.
\[
\bigvee_{X\in \{X\mid \forall a\in I_B.x_a=1\}} \left(\sum_{i=1}^n \sum_{j=1}^n a_{ij} x_i x_j\le q\right) 
\]
Our procedure uses a binary search to find the solution.  Because all elements in $A$ are greater than or equal to $-2(n-1)\cdot u$, the solution is in the set $[-2(n-1) \cdot u \cdot n^2, 0]$.  It follows that the solution can be obtained by querying the oracle $\log(2(n-1) \cdot u \cdot n^2) + O(1)$ times.  We have $\log(2(n-1) \cdot u \cdot n^2) + O(1) = \log(n-1) + 1 + \log(u) + 2\log(n) +O(1) \le O(\log M)$, because $u$ is a constant, and $n$ is bounded by $M$.   Therefore, our procedure can solve $\min_{X \in \{X\mid \forall a\in I_A. x_a=1\}} \sum_{i=1}^n \sum_{j=1}^n a_{ij} x_i x_j$ by querying the oracle for a problem in $\NP$ $O(\log M)$ times.

By the fact that the procedure solves the given UQUBO in polynomial time of $M$ by querying the oracle for the problem in $\NP$ $O(\log M)$ times where $M$ is the size of the quadratic expression of the given UQUBO, we have that UQUBO is in $\FP^{\NP[\log]}$. 

\end{proof}

Next, we determine that the hardness of LQUBO and UQUBO.  We formulate a subset of LQUBO, which is also a subset of UQUBO, referred to as SQUBO.  Then, we determine the complexity of SQUBO.  
\begin{definition}[SQUBO]
Let $X$ be a vector of binary variables, and $Q$ be an $n \times n$ upper triangular matrix of integers in which all elements are $-1$,$0$, or $1$.  SQUBO is the problem of minimizing the following quadratic  expression:
\[
\min_X \sum_{i=1}^n \sum_{j=1}^n q_{ij} x_i x_j
\]
\end{definition}

\begin{theorem}\label{thm:sqh}
SQUBO is ${\FP}^{\NP[\log]}$-hard.
\end{theorem}
\begin{proof}
We reduce the clique problem to SQUBO.  Recall that the clique problem (CLIQUE) is the problem of finding the size of the largest clique in a given undirected graph.  A clique is defined as follows.  Given an undirected graph $G=(V,E)$ where $V$ is a set of vertices and $E$ is a set of edges\footnote{Because the graph $G$ is an undirected graph, the set of edges $E$ is a set of sets with two vertices.},  a clique is a subset $C$ of $V$ such that every two distinct vertices in $C$ are connected by an edge of $G$.   As mentioned, CLIQUE is an ${\FP}^{\NP[\log]}$-complete problem, which was proven by Krentel~\cite{DBLP:journals/jcss/Krentel88}.

We demonstrate a method to transform the instance of CLIQUE to the instance of SQUBO in polynomial time in the size of the instance of CLIQUE such that the solution of the instance of CLIQUE is equal to the product of the solution of the instance of SQUBO and $-1$.   Note that the pair of the method to transform the instance of CLIQUE into the instance of SQUBO and the function that multiplies the input by $-1$ is the metric reduction from CLIQUE to SQUBO.

Given an undirected graph $G=(V,E)$, we construct the following quadratic binary expression:
\[
\mathrm{SQ}(X) = \mathrm{Obj}(X) + \mathrm{Cstr}(X),
\]
where 
\[
\begin{array}{rcl}
\mathrm{Obj}(X)&=& -\sum_{i\in\{1,\dots,n\}} x_{i} ^2\\
\mathrm{Cstr}(X)&=& \sum_{(v_i,v_j)\in \bar{E}} x_{i}x_{j}\\
V&=&\{v_1,\dots,v_n\}\\
\bar{E} &=& \{(v_i,v_j)\mid 1\le i<j\le n \wedge \{v_i,v_j\}\not\in E\}
\end{array}
\]
We provide the intuitive explanation of $\mathrm{SQ}(X)$.  The set of pairs of vertices $\bar{E}$ denotes the complement of the set of edges $E$.  The binary vector $X$ denotes the set of vertices, for example, $x_{i}$ represents the vertex $v_i$.  The term $\mathrm{Cstr}(X)$ adds a penalty to $\mathrm{SQ}(X)$ when $X$ does not satisfy the requirement of a clique.   The term $\mathrm{Obj}(X)$ maximizes the size of the set of vertices, which is represented by $X$.   Therefore, by minimizing $\mathrm{SQ}(X)$, we will obtain the size of the largest clique\footnote{Note that the set of vertices that the optimum assignment to X represents will not necessarily be the clique.  That is, there exist $\vec{x}'$ such that $\mathrm{SQ}(\vec{x}')$ is the minimum value and $\mathrm{Cstr}(\vec{x}')\not=0$.  However, we can construct a binary vector that represents the largest clique from $\vec{x}'$ in polynomial time, and is discussed below.}.  Note that the problem of finding the minimum value of $\mathrm{SQ}(X)$ is an instance of SQUBO with a $n\times n$ integer upper triangular matrix $Q$ defined as follows:
\[
q_{i j} =
\begin{cases}
-1& i=j\\
1& (v_i,v_j)\in \bar{E}\\
0& \text{otherwise}
\end{cases}
\]
Therefore, $\mathrm{SQ}(X)$ can be obtained in polynomial time in the size of the given instance of CLIQUE.  

We show that the minimum value of $\mathrm{SQ}(X)$ is equal to the product of the size of the largest clique in $G$ and $-1$.   We use the following proposition.
\begin{itemize}
\item If $\mathrm{SQ}(X)$ attains the minimum value when $X=x_1',\dots,x_n'$\\ then there exist $z_1, \dots, z_n$ such that $\mathrm{SQ}(x_1,\dots,x_n)=\mathrm{SQ}(z_1,\dots,z_n)$ and $\mathrm{Cstr}(z_1, \dots, z_n)=0$.
\end{itemize}

We prove this proposition.  Suppose that $\mathrm{SQ}(X)$ attains the minimum value when $X=x_1',\dots,x_n'$.  Let $z_1, \dots, z_n$ be a binary vector such that
\[
z_a = \begin{cases}
0&\sum_{i\in\{i\mid i<a\}}q_{i a}z_i + \sum_{j\in\{j\mid j>a\}} q_{a j}x_j'\ge 1\\
x_a'&\text{otherwise}
\end{cases}
\]
We show that 
\begin{align}
\textrm{SQ}(x_1',\dots,x_n')&=\textrm{SQ}(z_1,\dots,z_n)\textrm{,  and}\\ 
    \textrm{Cstr}(z_1,\dots,z_n)&=0
\end{align}

For the proposition (1), we use the following lemma: If $\mathrm{SQ}(X)$ attains the minimum value when $X=y_1,\dots,y_n$ then for any $a$ such that $y_a=1$, $\sum_{i\in\{i\mid i<a\}}q_{i a}y_i + \sum_{j\in\{j\mid j>a\}} q_{a j}y_j \le 1$.  We prove its contrapositive.  Suppose $y_1,\dots,y_n$ satisfies $\sum_{i\in\{i\mid i<a\}}q_{i a}y_i + \sum_{j\in\{j\mid j>a\}} q_{a j}y_j > 1$ for some $a$ such that $y_a=1$.  Let $y_1',\dots,y_n'$ be a binary vector such that $y_a'=0$ and $y_k=y_k'$ for $k\in\{k\mid k\not=a\}$.  Then, we have 
\[
\begin{array}{l}
\mathrm{SQ}(y_1,\dots,y_n)-\mathrm{SQ}(y_1',\dots,y_n')\\
\;\;= \sum_{i\in\{i\mid i<a\}}q_{i a}(y_i y_a - y_i' y_a') + \sum_{j\in\{j\mid j>a\}}q_{a j}(y_a y_j - y_a' y_j')\\
\;\;\;\;+ q_{a a} (y_a y_a - y_a' y_a')\\
\;\;= \sum_{i\in\{i\mid i<a\}}q_{i a}y_i + \sum_{j\in\{j\mid j>a\}}q_{a j} y_j -1\\
\;\;>0
\end{array}
\]
Therefore, we have $\mathrm{SQ}(y_1,\dots,y_n)>\mathrm{SQ}(y_1',\dots,y_n')$.  Thus, we prove the lemma.

Now, we prove the proposition (1). From the definition of $z_1,\dots,z_n$, for any $a\in[1,n]$, we have
\[
\begin{array}{l}
\textrm{SQ}(z_1,\dots,z_{a-1}, x_a',x_{a+1}',\dots,x_n')-\textrm{SQ}(z_1,\dots, z_{a-1},z_a,x_{a+1}',\dots,x_n')\\
\;\;=\sum_{i\in\{i\mid i<a\}}q_{i a}(z_i x_a' - z_i z_a) + \sum_{j\in\{j\mid j>a\}}q_{a j}(x_a' x_j' - z_a x_j')\\
\;\;\;\;+ q_{a a} (x_a' x_a' - z_a z_a)\\
\;\;=(x_a' - z_a)(\sum_{i\in\{i\mid i<a\}}q_{i a}z_i + \sum_{j\in\{j\mid j>a\}}q_{a j}x_j' + q_{a a})
\end{array}
\]
If $x_a'=z_a$ holds, trivially, the expression above is $0$.  Otherwise, it must be the case that $x_a'=1$ and $z_a=0$.  By the lemma above, we have $\sum_{i\in\{i\mid i<a\}}q_{i a}x_i' + \sum_{j\in\{j\mid j>a\}} q_{a j}x_j' \le 1$.  By the definition of $z_a$ and the fact that $x_a'\not=z_a$, we have 
\[
1\ge \sum_{i\in\{i\mid i<a\}}q_{i a}x_i' + \sum_{j\in\{j\mid j>a\}} q_{a j}x_j' \ge \sum_{i\in\{i\mid i<a\}}q_{i a}z_i + \sum_{j\in\{j\mid j>a\}} q_{a j}x_j' \ge 1
\]
It follows that $\sum_{i\in\{i\mid i<a\}}q_{i a}z_i + \sum_{j\in\{j\mid j>a\}}q_{a j}x_j' + q_{a a}=0$, and the expression above is $0$.  Thus, we have $\textrm{SQ}(z_1,\dots,z_{a-1}, x_a',x_{a+1}',\dots,x_n')=\textrm{SQ}(z_1,\dots, z_{a-1},z_a,x_{a+1}',\dots,x_n')$.  Therefore, by using this equation from $\mathrm{SQ}(x_1',\dots,x_n')$ to $\mathrm{SQ}(z_1,\dots,z_n)$, we can prove the proposition (1). 

Next, we prove the proposition (2) by contradiction.  Suppose that there exist $k$ and $\ell$ such that $q_{k\ell}z_k z_\ell=1$.  This implies that we have $q_{k \ell}z_k =1$, and hence $\sum_{i\in\{i\mid i<\ell\}}q_{i \ell}z_i \ge 1$.  By the definition of $z_\ell$, we have $z_\ell=0$.  This leads to a contradiction.  Thus, we prove the proposition (2).

By using the propositions above, we show that the minimum value of $\mathrm{SQ}(X)$ is equal to the product of the size of the largest clique in graph $G$ and $-1$.  Suppose that $\mathrm{SQ}(X)$ attains the minimum value when $X=\vec{x}'$.  In addition, suppose that $G$ has the largest clique $C$.  Trivially, we have $\mathrm{SQ}(\vec{x}')\le -|C|$, because when $\vec{z}$ represents the largest clique $C$ in $G$, we have $-|C|=\mathrm{SQ}(\vec{z})\ge \mathrm{SQ}(\vec{x}')$ by the definition of $\mathrm{SQ}$.  Next, we show $\mathrm{SQ}(\vec{x}')\ge -|C|$ by contradiction.  Suppose that $\mathrm{SQ}(\vec{x}')<-|C|$.  By the proposition above, there exists $\vec{z}$ such that $\mathrm{SQ}(\vec{x}')=\mathrm{SQ}(\vec{z})$ and $\textrm{Cstr}(\vec{z})=0$.  Then, we have $C'=\{v_a\mid z_a=1\}$ is a clique of $G$, since for any $v, v'\in C'$, $\{v, v'\}\in E$ holds.  However, we have $-|C'| = \mathrm{SQ}(\vec{z})=\mathrm{SQ}(\vec{x}')  <- |C|$.  This leads to a contradiction with the fact that $C$ is the largest clique.  Thus, we have $\mathrm{SQ}(\vec{x})\ge -|C|$, and hence we have $\mathrm{SQ}(\vec{x})=-|C|$.
\end{proof}

We thus have the following completeness results of LQUBO and UQUBO.  
\begin{theorem}\label{thm:lqc}
LQUBO is ${\FP}^{{\NP}[\log]}$-complete.
\end{theorem}
\begin{theorem}\label{thm:uqc}
UQUBO is ${\FP}^{{\NP}[\log]}$-complete.
\end{theorem}
\noindent
Theorem~\ref{thm:lqc} follows from Theorem~\ref{thm:lqin}, Theorem~\ref{thm:sqh}, and the fact that every instance of SQUBO is also an instance of LQUBO.  Theorem~\ref{thm:uqc} follows from Theorem~\ref{thm:uqin}, Theorem~\ref{thm:sqh}, and the fact that every instance of SQUBO is also an instance of UQUBO.  

\section{Hardness of decision version of QUBO}
\label{sec:dqubo}
In this section, we demonstrate the computational complexity of the decision version of QUBO problem, referred to as DQUBO.  We define the problem as follows.
\begin{definition}[DQUBO]
Let $X$ be a vector of binary variables, $Q$ be an $n\times n$ upper triangular matrix of integers, and $q$ be an integer.  DQUBO is the problem of deciding whether the following equality holds:
\[
\min_X \sum_{i=1}^n \sum_{j=1}^n q_{ij} x_i x_j = q 
\]
\end{definition}
Intuitively, DQUBO checks whether the solution of the instance of QUBO problem is equal to a given integer.  

We prove that DQUBO is a ${\DP}$-complete problem.  The computational complexity class ${\DP}$ is
the set of decision problems such that ${\DP}=\{P_0\cap P_1\mid P_0\in{\NP}, P_1\in{\coNP}\}$, which is defined by Papadimitriou and Yannakakis~\cite{PAPADIMITRIOU1984244}. 

To prove that DQUBO is ${\DP}$-complete, we prove that DQUBO is ${\DP}$-hard and that DQUBO is in ${\DP}$ as follows.
\begin{lemma}
\label{lem:dqh}
DQUBO is $\DP$-hard.
\end{lemma}
\begin{proof}
We show a polynomial-time many-one reduction from the exact clique problem to DQUBO.  The exact clique problem (ECLIQUE) is the problem of deciding whether the size of the largest clique in a given undirected graph $G$ is equal to a given integer $k$.  Note that ECLIQUE is DP-complete~\cite{PAPADIMITRIOU1984244}.  

In Theorem~\ref{thm:sqh}, we showed the quadratic binary expression $\mathrm{SQ}(X)$ such that the size of the largest clique in a given undirected graph is equal to the product of the minimum value of $\mathrm{SQ}(X)$ and $-1$.  We also showed that the expression $\mathrm{SQ}(X)$ can be constructed in polynomial time in the size of a given undirected graph.  By using the expression, we can check whether $(G, k)$ is a solution of ECLIQUE by checking whether $(\mathrm{SQ}(X),-k)$ is a solution of DQUBO.  
\end{proof}

\begin{lemma}
\label{lem:dqin}
DQUBO is in DP.
\end{lemma}
\begin{proof}
We construct the problem $P_1$ such that ${\text{DQUBO}}={\text{DLEQUBO}} \cap P_1$ and $P_1\in{\coNP}$. 
Let $P_1$ is the problem of deciding whether the following inequality holds:
\[
\min_X \sum_{i=1}^n \sum_{j=1}^n q_{ij} x_i x_j > q-1 
\]
where $X$ is a vector of binary variables, $Q$ is an $n \times n$ upper triangular matrix of integers, and $q$ is an integer.   Trivially, we have ${\text{DQUBO}}={\text{DLEQUBO}} \cap P_1$.  We also have
\[
\min_X \sum_{i=1}^n \sum_{j=1}^n q_{ij} x_i x_j > q-1 \Leftrightarrow \neg \bigvee_X \left(\sum_{i=1}^n \sum_{j=1}^n q_{ij} x_i x_j\le q-1\right) 
\]
The problem of deciding whether the disjunction above holds is in ${\NP}$, because it is solvable by a non-deterministic Turing machine in polynomial time.  Therefore, we have $P_1$ is in ${\coNP}$.  

Therefore, the theorem follows from the fact that ${\text{DQUBO}}={\text{DLEQUBO}} \cap P_1$, ${\text{DLEQUBO}}\in {\NP}$ (Lemma~\ref{lem:DLEQUBO}), and $P_1\in{\coNP}$. 
\end{proof}
We thus have the following completeness result.
\begin{theorem}
DQUBO is DP-complete.
\end{theorem}

\section{Discussion}\label{sec:dis}
\subsection{Hardness of QUBO with rational coefficients}\label{ssec:hqrc}
The ${\FP}^{\NP}$-completeness result of QUBO with integer coefficients (Theorem~\ref{thm:qc}) also holds for QUBO with rational coefficients, if we can assume that a rational number is represented by a pair of two integers: a numerator and denominator.   Trivially, the ${\FP}^{\NP}$-hardness result (Theorem~\ref{thm:qfpnph}) also holds for QUBO with rational coefficients.  We show that the ${\FP}^{\NP}$-membership result (Theorem~\ref{thm:qfpnp}) also holds for QUBO with rational coefficients.  We can obtain a quadratic binary expression in which all coefficients are integers by multiplying each numerator by a product of all the denominators of the coefficients.  The size of the obtained quadratic binary expression is bounded by $O(M^2)$, where $M$ is the size of the original expression with rational coefficients.   By Theorem~\ref{thm:qfpnp}, we can obtain the solution of QUBO with integer coefficients in polynomial time.  In addition, we can obtain the solution of the original expression with rational coefficients by dividing the solution of QUBO with integer coefficients by the product of the denominators. 

\subsection{Application}\label{ssec:app}
 The ${\FP}^{\NP}$-hardness result (Theorem~\ref{thm:qfpnph}) implies that a problem that is not harder than ${\FP}^{\NP}$ can be solved via QUBO such as problems in ${\NP}$ and ${\FP}^{\NP}$.   SAT is an ${\NP}$-complete problem that has many applications such as planning systems~\cite{DBLP:conf/ijcai/KautzS99}, software verification~\cite{DBLP:conf/tacas/BiereCCZ99}, and hardware verification~\cite{DBLP:conf/cav/KaivolaGNTWPSTFRN09}.   A concrete reduction from SAT to QUBO was proposed by Lucas~\cite{DBLP:journals/corr/abs-1302-5843}.   There are also ${\FP}^{\NP}$-complete problems such as the knapsack problem.  The knapsack problem is a problem, given integers $m_1, \dots, m_n$ and $K$, to find the maximum value $\sum_{m\in S} m$ such that $S\subseteq \{m_1, \dots, m_n\}$ and $\sum_{m\in S} m < K$.  Krentel~\cite{DBLP:journals/jcss/Krentel88} proved that the knapsack problem is an ${\FP}^{\NP}$-complete problem.  The knapsack problem has many applications such as resource allocation and finance portfolio optimization~\cite{DBLP:books/daglib/0010031}.  Reduction from the knapsack problem to 01IP is trivial; therefore, we can reduce the knapsack problem to QUBO by Theorem~\ref{thm:qfpnph}.  Applications of the knapsack problem are thus also solvable via QUBO.

The ${\FP}^{\NP}$-membership result of QUBO (Theorem~\ref{thm:qfpnp}) implies that a problem that is harder than $\FP^\NP$ can hardly be solved via QUBO.  For example,  2QBF is a problem that is harder than ${\FP}^{\NP}$.  Recall that this problem is the problem of checking whether a given quantified Boolean formula $\exists \vec{x} \forall\vec{y} .\phi$ where $\phi$ is a quantifier free Boolean formula can be satisfied.  2QBF is an ${\NP}^{\NP}$-complete problem.   We show that 2QBF can hardly be reduced to QUBO.  Suppose that 2QBF can be reduced to QUBO with respect to a metric reduction.  It follows that every problem that is solvable by a deterministic polynomial-time Turing machine that can query an oracle for 2QBF is solvable by a deterministic polynomial-time Turing machine that can query an oracle for QUBO, because we can calculate the answer of the oracle for 2QBF in polynomial time by querying the oracle for QUBO.  That is, ${\PO}^{\mathrm{2QBF}}\subseteq {\PO}^{\mathrm{QUBO}}$, thus ${\PO}^{{\NP}^{\NP}}\subseteq {\PO}^{{\FP}^{\NP}}$\footnote{Here, $A\subseteq B$ signifies that any problem in $A$ can be reduced to $B$ with respect to a many-one reduction.}.  ${\PO}^{{\FP}^{\NP}}$ is a class of decision problems that are solvable in polynomial time by a Turing machine with an oracle for ${\FP}^{\NP}$.  We have ${\PO}^{{\FP}^{\NP}}\subseteq {\PO}^{\NP}$, because we can calculate the answer of the oracle for ${\FP}^{\NP}$ in polynomial time by querying an oracle for ${\NP}$.  The fact that ${\PO}^{{\NP}^{\NP}}\subseteq {\PO}^{\NP}$ implies the collapse of the polynomial hierarchy to level 2.  However, the polynomial hierarchy is believed not to collapse.  Therefore, 2QBF can hardly be solved via QUBO.  

\section{Related work}\label{sec:rw}

The computational complexity of combinatorial optimization problems has been examined in various studies~\cite{DBLP:journals/jcss/Krentel88, DBLP:journals/mst/GasarchKR95}.   These studies used metric reduction to demonstrate the computational complexity of combinatorial optimization problems, and in this paper, we adopt their approach to demonstrate the computational complexity of QUBO.  

The computational complexity of QUBO was studied by Pardalos and Jha~\cite{DBLP:journals/orl/PardalosJ92}, who proved that quadratic 0-1 programming is ${\NP}$-hard.  Note that the definition of quadratic 0-1 programming is slightly different from the definition of QUBO.  
Specifically, quadratic 0-1 programming can be defined as QUBO with a symmetric matrix of integers.  Therefore, the ${\NP}$-hardness of quadratic 0-1 programming implies that the ${\NP}$-hardness of the QUBO problem, because we can trivially reduce quadratic 0-1 programming into QUBO.  

A number of QUBO formulations of $\NP$-hard problems have been presented by researchers~\cite{DBLP:journals/corr/abs-1302-5843,DBLP:journals/4or/GloverKD19, DBLP:journals/vlsisp/ChapuisDHR19}, who demonstrated the method to solve $\NP$-complete problems, ${\FP}^{{\NP}[\log]}$-complete problems, and ${\FP}^{\NP}$-complete problems via QUBO.   Some of their work can be used to prove that QUBO are hard.  However, their work does not imply the completeness results of QUBO.

Whereas we demonstrated how to solve QUBO by using Turing machine with an oracle to prove the computational complexity of QUBO, other researchers studied exact and heuristic methods to solve QUBO.   We refer readers to an existing survey~\cite{Kochenberger2014TheUB}.

\section{Conclusion}\label{sec:con}
In this paper, we investigated the hardness of QUBO using computational complexity theory.  We showed that QUBO with integer coefficients is an ${\FP}^{\NP}$-complete problem.  This result implies that every problem in ${\FP}^{\NP}$ (e.g., the traveling salesman problem, the knapsack problem) can be reduced to QUBO with integer coefficients, thus receiving a benefit from quantum annealing.   The completeness result also implies that a problem that is harder than ${\FP}^{\NP}$ cannot be solved via QUBO.  We also showed that QUBO with a constant lower bound for the coefficients and QUBO with a constant upper bound for the coefficients are ${\FP}^{{\NP}[\log]}$-complete problems, and that the decision version of QUBO is a ${\DP}$-complete problem.  

\section*{Acknowledgements}
We want to thank Akira Miki and Chih-Hong Cheng for useful advice.

\bibliographystyle{plain}
\bibliography{references.bib}

\end{document}